\documentclass{icm2010}

\usepackage{amsmath,amssymb}
\usepackage{bbm}
\usepackage{graphics,graphicx,epsfig,psfrag}  

\def\EE{\mathbb{E}}

\def\RR{\mathbb{R}}

\def\hZ{\widehat{Z}}

\newcommand{\btimes}{\mbox{\huge \raisebox{-0.2ex}{$\times$}}} 
\newcommand{\one}{\mathbbm 1}

\newcommand{\ts}{\hspace{0.5pt}}

\def\cM{\mathcal{M}}

\title[Deterministic and stochastic aspects of single-crossover
recombination]{Deterministic
  and stochastic aspects \\ of  single-crossover recombination}
\author[Baake]{Ellen Baake}
\contact[ebaake@techfak.uni-bielefeld.de]{Faculty of Technology, 
Bielefeld University, 33594 Bielefeld, Germany}

\newtheorem{theorem}{Theorem}
\newtheorem{prop}[theorem]{Proposition}
\newtheorem{lemma}[theorem]{Lemma}

\theoremstyle{definition}

\begin{document}

\begin{abstract}

This contribution is concerned with mathematical 
models for the dynamics of the genetic 
composition of populations evolving
under recombination. Recombination is the genetic
mechanism by  which two parent individuals create the
mixed type of their offspring during sexual reproduction.
The corresponding models are large, nonlinear
dynamical systems (for the deterministic treatment that applies
in the infinite-population limit), or
interacting particle systems (for the stochastic treatment
required for finite populations). We review recent progress on
these  difficult problems. In particular, we
present a closed solution of the deterministic continuous-time
system, for the important special case of single crossovers;
we extract an underlying linearity; we analyse how this carries
over to the corresponding stochastic setting; and we provide a 
solution of the analogous deterministic discrete-time dynamics, in terms of 
its generalised eigenvalues and a simple recursion for the corresponding
coefficients. 
 \end{abstract}

\begin{classification}
Primary 92D10, 
        34L30;  
Secondary 37N25, 
          06A07,  
          60J25.  
\end{classification}

\begin{keywords}
Population genetics, recombination dynamics, M\"obius linearisation and
diagonalisation, correlation functions, Moran model.
\end{keywords}

\maketitle


\section{Introduction}
Biological evolution is a complex phenomenon driven by various processes,
such as muation and recombination of genetic material, reproduction of 
individuals, and selection of favourable types. The area of
\emph{population genetics} is concerned with how these processes
shape and change the genetic structure of populations. 
\emph{Mathematical population genetics} was founded in the 1920's by
Ronald Fisher, Sewall Wright, and John Haldane, and thus is among the
oldest areas of mathematical biology. The reason for its continuing (and 
actually 
increasing) attractiveness for both mathematicians and biologists
is at least twofold: Firstly, there is a true need for
mathematical models and methods, since the outcome of evolution
is impossible to predict (and, thus, today's genetic data are
impossible to analyse) without their help. Second, the processes of
genetics lend themselves most naturally to a mathematical formulation
and give rise to a wealth of fascinating new problems, concepts, and
methods.

This contribution will focus on the phenomenon of \emph{recombination},
in which two parent individuals are involved in creating the
mixed type of their offspring during sexual reproduction. 
The essence of this process is illustrated in Fig.~\ref{fig:lifecycle}
and may be
idealised and summarised as follows. 

\begin{figure}[h]
  \begin{center}
  \includegraphics[scale=.34]{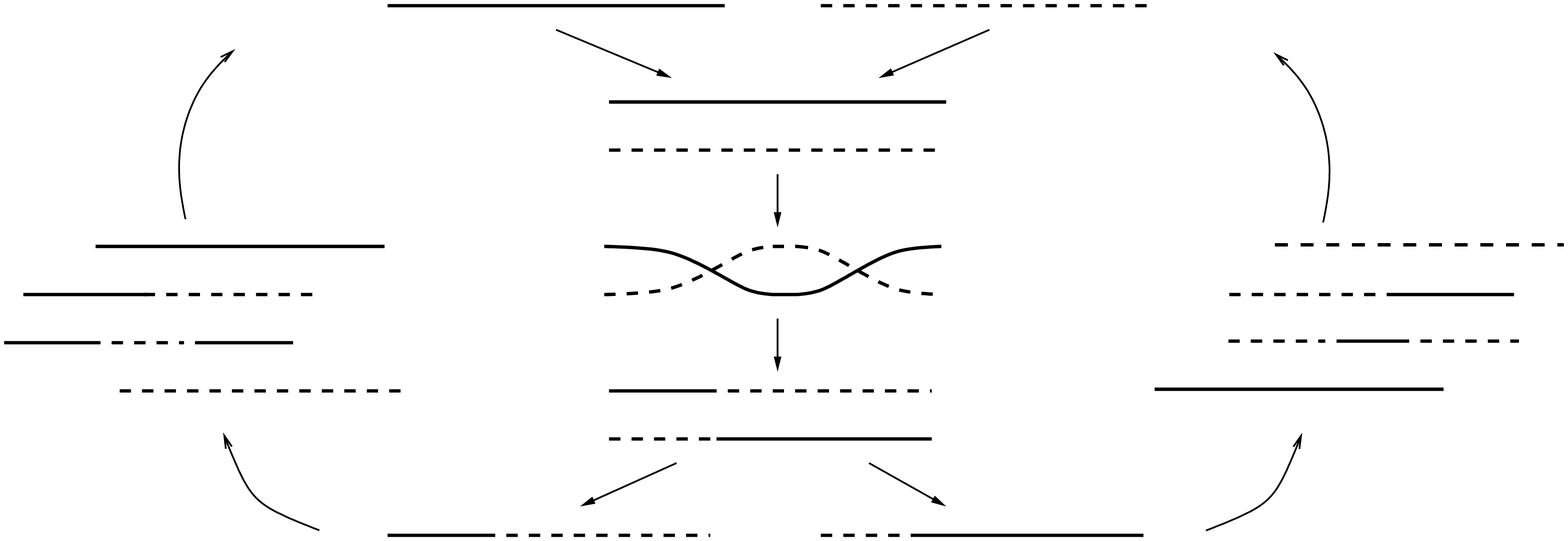}
  \end{center}  \caption{\label{fig:lifecycle} {\small Life cycle of
   a population under sexual reproduction and recombination. Each line
symbolises a sequence of sites that defines a gamete (like the two
at the top that start the cycle as `egg' and `sperm'). The pool of gametes
at the left and the right comes from a large population of recombining
individuals. These sequences meet randomly to start the next round of the
cycle.}}
\end{figure}

Genetic information is encoded
in terms of sequences of finite length. Eggs and  sperm 
(i.e.,  female and  male germ cells or \emph{gametes}) 
each carry a single copy of such a sequence. They go through the
following life cycle: At fertilisation, 
two gametes meet randomly and unite,
thus starting the life of a new individual, which is equipped with both the
maternal and the paternal sequence. At maturity, this individual
will generate its own germ cells. This process includes recombination,
that is, the maternal and paternal sequences perform one or more
\emph{crossovers} and are cut and relinked accordingly, so that
two `mixed' sequences emerge. These are the new gametes and start
the next round of  fertilisation 
(by random mating within a large population).

Models of this process aim at describing the dynamics of the
genetic composition of
a population that goes through  this life cycle repeatedly.
These models  come in various flavours: in discrete or
continous time; with various assumptions about the crossover pattern;
and, most importantly, in a deterministic or a stochastic
formulation, depending on whether the population is assumed to be
so large that stochastic fluctuations may be neglected. 
In any case, however, the resulting process appears difficult
to treat, due to the large number of possible states and the
nonlinearity generated by the random mixture of gametes.  
Nevertheless, a number of solution procedures have been 
discovered for the deterministic discrete-time setting 
\cite{Benn54,Daws02,Geir44}, and 
the underlying mathematical structures were investigated
within the framework of genetic algebras, see \cite{Lyub92,McHR83,Ring85}.
Quite generally, the 
solution relies on a certain nonlinear transformation (known as Haldane
linearisation) from (gamete or type)  frequencies 
to suitable correlation functions, which 
decouple from each other and decay geometrically. But if sequences of
more than three sites 
are involved, this transformation must be constructed via recursions that 
involve the parameters of the recombination process, and is  not 
available explicitly except in the trivial case of independent sites. For a 
review of the area, see \cite[Ch.~V.4]{Buer00}.

In this contribution, we  concentrate on  a  
special case that is both biologically and mathematically relevant,
namely, the situation in which at most one crossover
happens at any given time. That is,
only recombination events  may occur that partition the 
sites of a sequence  into two  parts 
that correspond to the sites before and after a given crossover point.
We analyse the resulting models in continuous time (both deterministic
and stochastic), as well as in discrete time. For the deterministic
continuous-time system (Section~\ref{sec:det_cont}), a simple explicit
solution can be given. This simplicity is due to some underlying linearity;
actually, the system may even be diagonalised
(via a \emph{nonlinear} transformation). 
In Section~\ref{sec:stoch_cont}, we
consider the corresponding stochastic process (still in
continuous time), namely, the Moran model with recombination.
This also takes into account the
resampling effect that comes about via random reproduction in a
finite population. In particular, we investigate the relationship
between the expectation of the Moran model and the solution
of the deterministic continuous-time model.
We finally tackle deterministic single-crossover
dynamics in \emph{discrete} time (Section~\ref{sec:det_discr}).
This setting implies additional dependencies, which become
particularly transparent when the so-called
\emph{ancestral recombination process} is considered.  A solution may still
be given, but  its  coefficients  must be determined
recursively. 

Altogether, it will turn out that
the corresponding models, and their analysis, have 
various mathematical facettes that are intertwined with each other,
such as differential equations,
probability theory, and combinatorics.

\section{Deterministic dynamics, continuous time}
\label{sec:det_cont}
\subsection{The model}
We  describe populations at the level of their gametes and thus
identify gametes with individuals.
Their genetic information  is encoded in terms of a 
linear arrangement of sites, indexed by the set $S:= \{ 0,1, \dots , n\}$.
For 
each site $i\in S$, there is a set $X_i$ of `letters'  that may possibly 
occur at that site. 
To allow for a convenient  notation, we restrict ourselves  to 
the simple but important case of \emph{finite} sets $X_i$; for the
full generality of arbitrary locally compact spaces $X_i$, the reader is
referred to  \cite{Baak05} and \cite{BaBa03}.

A \emph{type} is thus  defined 
as a sequence 
\[
x=(x^{}_0,x^{}_1, \dots , x^{}_n) \in X_0 \times X_1 \times \cdots 
\times X_n =:X,
\]
where $X$ is called the \emph{type space}. By 
construction, $x^{}_i$ is the $i$-th coordinate of $x$, and we define 
$x^{}_I := (x^{}_i)_{i \in I}$ as the collection of coordinates 
with indices in $I$, where $I$ is a subset of $S$. 
A \emph{population} is identified with  a non-negative measure
$\omega$ on $X$.
Namely, $\omega(\{x\})$ denotes the frequency of individuals of type $x \in X$ 
and
$\omega(A) := \sum_{x \in A} \omega(\{x\})$ for $A \subseteq X$;
we abbreviate $\omega(\{x\})$ as $\omega(x)$. The set of all
nonnegative measures on $X$ is denoted by $\cM_{\geqslant 0}(X)$.
If we define $\delta_x$ as the point measure on $x$ 
(i.e., $\delta_x(y)= \delta_{x,y}$ for $x,y \in X$), we can also write
$\omega = \sum_{x \in X} \omega(x) \delta_x$. We may,
alternatively, interpret
$\delta_x$ as the basis vector of $\RR_{\geq 0}^{\lvert X \rvert}$
that corresponds to $x$ (where a suitable ordering of types is implied,
and $\lvert X \rvert$ is the number of elements in $X$);
$\omega$ is thus
identified with a vector in $\RR_{\geq 0}^{\lvert X \rvert}$.

At this stage, frequencies need not be normalised;
$\omega(x)$ may simply be thought of as the size of the subpopulation
of type $x$, measured in units so large that it may be considered a
continuous quantity. 
The corresponding normalised version $p := \omega/\|\omega\|$ 
(where $\|\omega\|
:= \sum_{x \in X} \omega(x)=\omega(X)$ is the total population size)
is then a probability distribution on $X$, and may be identified
with a probability vector.

Recombination acts on the links 
between the sites; the links are collected
into the set $L := \left\{ \frac{1}{2} , 
\frac{3}{2}, \dots , \frac{2n-1}{2} \right\}$. We shall use Latin indices for 
the sites and Greek indices for the links, and the implicit rule will always 
be that $\alpha = \frac{2i+1}{2}$ is the link between sites $i$ and $i+1$;
see Figure~\ref{fig:sitesandlinks}.

\begin{figure}[h]
  \psfrag{0}{$0$}
  \psfrag{1}{$1$}
  \psfrag{n}{$n$}
  \psfrag{iinS}{$i \in S$}
  \psfrag{12}{$\frac{1}{2}$}
  \psfrag{32}{$\frac{3}{2}$}
  \psfrag{ungerade}{\hspace{1mm}$\frac{2n-1}{2}$}
  \psfrag{alpha}{$\alpha \in L$}
  \begin{center}
  \includegraphics{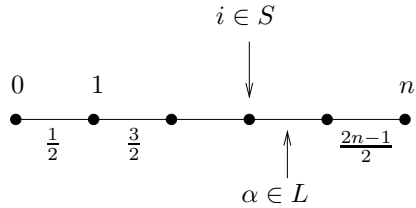}
  \end{center}  \caption{\label{fig:sitesandlinks} {\small Sites and links.}}
\end{figure}

Let recombination happen in every individual, and at every link 
$\alpha \in L$, 
at rate $\varrho^{}_\alpha > 0$.
More precisely, for every $\alpha \in L$,
every individual exchanges, at rate $\varrho^{}_{\alpha}/2$, 
the sites after link $\alpha$ with those of a 
randomly chosen partner. Explicitly, if
the `active' and the partner individual are of types $x$ and $y$,
then the new pair has types $(x^{}_0,x^{}_1,\ldots,
x^{}_{\lfloor \alpha \rfloor}, y^{}_{\lceil \alpha \rceil}, \ldots, y^{}_n)$
and $(y^{}_0,y^{}_1,\ldots,
y^{}_{\lfloor \alpha \rfloor}, x^{}_{\lceil \alpha \rceil}, \ldots, x^{}_n)$,
where $\lfloor\alpha\rfloor 
(\lceil\alpha\rceil )$ is the largest integer below $\alpha$ (the 
smallest above $\alpha$);
see Fig.~\ref{fig:reco}. 
Since every individual
can occur as either the `active' individual or as its randomly
chosen partner, we have a total rate of  $\varrho^{}_\alpha$ for 
crossovers at link $\alpha$. 
For later use, we also define
$\varrho := \sum_{\alpha \in L} \varrho^{}_{\alpha}$.

\begin{figure}[h]
\psfrag{x}{$x^{}_0, \ldots , x^{}_n$}
\psfrag{y}{$y^{}_0, \ldots , y^{}_n$}
\psfrag{r}{$\varrho^{}_{\alpha}$}
\psfrag{xy}{$x^{}_0,\ldots,x^{}_{\lfloor \alpha \rfloor}, y^{}_{\lceil \alpha \rceil},\ldots, y^{}_n$}
\psfrag{yx}{$y^{}_0,\ldots,y^{}_{\lfloor \alpha \rfloor}, x^{}_{\lceil \alpha \rceil},\ldots, x^{}_n$}
\psfrag{*}{$* \;\,, \ldots , \; *$}
\psfrag{x*}{$x^{}_0,\ldots,x^{}_{\lfloor \alpha \rfloor}, 
            *\; , \; \ldots \; ,\; *$}
\psfrag{*x}{$* \;,\; \ldots \;, \; * \;, \; x^{}_{\lceil \alpha \rceil}, \ldots , x^{}_n$}
\begin{center}
\hspace{-2.5cm}\includegraphics{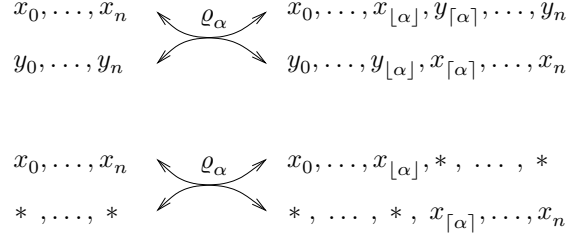}
\end{center}
\caption{\label{fig:reco} {\small 
Upper panel: Recombination between individuals of 
type $x$ and $y$. Lower panel:
The corresponding `marginalised' version that summarises all events
by which individuals of type $x$ are gained or lost (a `$*$' at site $i$
stands for an arbitrary element of $X_i$).
Note that, in either case, the process can go both ways, as
indicated by the arrows.}}
\end{figure}



In order to formulate the corresponding model, let us 
introduce the projection operators 
$\pi^{}_i$, $i \in S$,  via
\begin{equation}\label{eq:projector}
\begin{array}{rlcl}
\pi^{}_i : & X_0 \times X_1 \times \dots \times X_n & \longrightarrow & X_i\\
& (x^{}_0, x^{}_1, \dots , x^{}_n) & \mapsto & x^{}_i,
\end{array}
\end{equation}
i.e., $\pi^{}_i$ is the canonical projection to the $i$-th coordinate.
Likewise, for any index set $I \subseteq S$, one defines a projector
\[
\begin{array}{rlcl}
\pi^{}_I: & X & \longrightarrow &  \btimes_{i \in I} X_i =: X_I\\
& (x^{}_0,x^{}_1,\ldots, x^{}_n) & \mapsto & (x^{}_i)^{}_{i \in I} =: x^{}_I.
\end{array}
\] 
We shall frequently use the abbreviations 
$\pi^{}_{<\alpha} := \pi^{}_{\{1,\ldots,\lfloor \alpha \rfloor\}}$ and
$\pi^{}_{>\alpha} := \pi^{}_{\{\lceil \alpha \rceil, \ldots, n\}}$,
as well as $x^{}_{<\alpha}:= \pi^{}_{<\alpha}(x)$, 
$x^{}_{>\alpha}:= \pi^{}_{>\alpha}(x)$. The projectors 
$\pi^{}_{<\alpha}$ and $\pi^{}_{>\alpha}$
may be thought of as \emph{cut and forget} operators 
because they take the leading or trailing segment of a sequence $x$, and 
forget about the rest. 

Whereas the $\pi^{}_I$  act on the types, 
we also need the induced mapping at the level of the population, namely,
\begin{equation} \label{eq:marg}
\begin{array}{rccl}
\pi^{}_I. : & \cM_{\geqslant 0} & 
           \longrightarrow & \cM_{\geqslant 0} \\
         &  \omega     & \mapsto &  
              \omega \circ \pi^{-1}_I =: \pi^{}_I . \omega,
\end{array}
\end{equation}
where $\pi^{-1}_I$ denotes the preimage under 
$\pi^{}_I$. 
The operation $.$ (where the dot is on the line) is  the `pullback' 
of $\pi^{}_I$ w.r.t. 
$\omega$;
so, $\pi^{}_I . \omega$ is
the marginal distribution of $\omega$ with respect
to the sites in $I$. In particular,  
$(\pi^{}_{<\alpha} . \omega)(x^{}_{<\alpha})$
is  the marginal frequency  of  sequences prescribed at 
the sites before $\alpha$, and vice versa for the sites after $\alpha$.

Now, single-crossover
recombination (at the level of the population) means the relinking of 
a randomly chosen leading segment with a randomly chosen trailing segment.
We therefore introduce (elementary) recombination operators 
(or \emph{recombinators}, for short), $R_\alpha : \cM_{\geqslant 0} \to 
\cM_{\geqslant 0}$ for $\alpha \in L$, defined by 
\begin{equation}\label{eq:recombinator}
R_\alpha (\omega) := \frac{1}{\|\omega\|}
\big ( (\pi^{}_{<\alpha} . \omega) \otimes (\pi^{}_{>\alpha} . \omega) \big ).
\end{equation}
Here, the tensor product reflects the independent combination (i.e., the 
product measure) of the two marginals $\pi^{}_{<\alpha} . \omega$ and 
$\pi^{}_{>\alpha} . \omega$. $R_\alpha$ is therefore  a \emph{cut and relink}
operator. $R_\alpha (\omega)$ may be understood as the population 
that emerges if \emph{all} individuals of the population $\omega$ 
disintegrate into 
their leading and trailing segments and these are relinked randomly.
Note that $\|R_\alpha (\omega)\| = \|\omega\|$.

The recombination dynamics  may thus be compactly 
written as
\begin{equation} \label{eq:ode}
\dot{\omega}_t = \sum_{\alpha\in L} 
\varrho^{}_\alpha \big ( R_\alpha (\omega_t) - \omega_t \big ) = 
\sum_{\alpha\in L} \varrho^{}_\alpha (R_\alpha -\one) (\omega_t)
=: \varPhi(\omega_t),
\end{equation}
where $\one$ is the identity operator. Note that \eqref{eq:ode}
is a large, nonlinear system of ordinary differential equations (ODEs).

\subsection{Solution of the ODE system}

The solution of (\ref{eq:ode}) relies on some elementary properties of 
our recombinators. Most importantly, they are idempotents and commute
with each other, i.e.,
\begin{eqnarray} 
R^2_\alpha & = & R_\alpha , \qquad\  \alpha \in L, \label{eq:idpot}\\
R_\alpha R_\beta & = & R_\beta R_\alpha,  \quad \alpha , \beta \in L. 
\label{eq:commut}
\end{eqnarray}
These properties are intuitively  plausible: if the links
before $\alpha$ are already independent of those after $\alpha$
due to a previous recombination event, then further recombination 
at that link does not change the situation; and if a product measure is 
formed with respect to two links $\alpha$ and $\beta$, the result does not 
depend on the order in which the links are affected. For the proof, we 
refer  to \cite[Prop.\ 2]{BaBa03}; let us only 
mention here that it relies on the elementary fact that, for $\omega\in 
\cM_{\geqslant 0}$,
\[
\begin{array}{ll}
\pi^{}_{<\alpha} . \big (R_\beta (\omega) \big ) = \pi^{}_{<\alpha} . \omega, \quad & 
\mbox{ for } \beta  \geq \alpha , \mbox{ and} \\
\pi^{}_{>\alpha} . \big (R_\beta (\omega) \big ) = \pi^{}_{>\alpha} . \omega, & \mbox{ for } 
\beta \leq \alpha ;
\end{array}
\]
that is, recombination at or after $\alpha$ does not affect 
the marginal frequencies
at sites before $\alpha$, and vice versa.

We now define \emph{composite} recombinators as
\[
R_G := \prod_{\alpha\in G} R_\alpha \qquad \mbox{ for } G\subseteq L.
\]
Here, the product is to be read as composition; it is, indeed, a
product if the recombinators are written in terms of their multilinear
matrix
representations, which is available in the case of finite types
considered here (see \cite{Baak01}).
By property (\ref{eq:commut}), the order in the composition plays no role.
Furthermore, (\ref{eq:idpot}) and (\ref{eq:commut}) obviously entail 
$R_G R_H = R_{G\cup H}$ for $G,H \subseteq L$.

With this in hand, we can now state an explicit solution of our problem,
namely,

\begin{theorem}\label{thm:sol}
The solution of the single-crossover dynamics 
\eqref{eq:ode} with initial value $\omega_0$  can be given in
closed form as
\begin{equation}\label{eq:recosol}
\omega_t = \sum_{G\subseteq L} a_G^{} (t) R_G (\omega_0) =: \varphi_t(\omega_0)
\end{equation}
with coefficient functions
\[
   a_G^{} (t) = \prod_{\alpha\in L\backslash G} 
e^{-\varrho^{}_\alpha t} \prod_{\beta\in G} (1-e^{-\varrho^{}_\beta t});
\] 
i.e., $\varphi_t$ is the {\em semigroup}
belonging to the recombination equation \eqref{eq:ode}. \qed
\end{theorem}

For the proof,  the reader is referred to
\cite[Thm.\ 2]{BaBa03} or \cite[Thm.\ 3]{Baak05} (the former article
contains the original, the latter 
a shorter and more elegant version of the proof). Let us note  that 
the coefficient functions  can be interpreted probabilistically. Given an
individual sequence in the population, $a^{}_{G} (t)$ is the 
probability that the set of links that have seen at least one crossover 
event until time $t$ is precisely the set $G$ (obviously, $\sum_{G \subseteq L }
a^{}_{G} (t)=1$). Note that the product structure of the $a^{}_{G} (t)$
implies independence of links, a decisive feature of the single-crossover
dynamics in continuous time, as we shall see later on. Note also that,
as $t \to \infty$, $\omega_t$ converges to the stationary state
\begin{equation}\label{eq:stat_state}
\omega_{\infty} = 
\frac{1}{\|\omega_0\|^{n-1}}
  \bigotimes_{i =1}^{n} (\pi^{}_{i} . \omega_0),
\end{equation}
in which all sites are independent.

\subsection{Underlying linearity}
The simplicity of the solution in Theorem~\ref{thm:sol}
comes as a certain surprise. After all, 
explicit solutions to  large, nonlinear  ODE systems are  rare
 -- they are usually available for linear systems
at best. For this reason,  the recombination equation
and its solution have already been taken up in the framework of
functional analysis, where they have led to an extension of potential
theory \cite{Popa07}. We will now show that there
is an underlying linear structure that is hidden behind the solution.
It can be stated as 
follows, compare~\cite[Sec.~3.2]{BaBa03} for details.

\begin{theorem} \label{thm:linear}
    Let $\bigl\{c_{G'}^{(L')}(t) \mid \varnothing\subseteq G'\subseteq L'\subseteq L\bigr\}$ 
    be a family of non-negative functions with 
    $c_G^{(L)}(t) = c_{G_1}^{(L_1)}(t)\, c_{G_2}^{(L_2)}(t)$, valid for any 
    partition $L=L_1\ts\dot{\cup}\ts L_2$ of the set $L$ and all $t\geq 0$,
    where $G^{}_i := G\cap L^{}_i$.
    Assume further that these
    functions satisfy $\sum_{H\subseteq L'} c_{H}^{(L')} (t) = 1$ for any 
    $L'\subseteq L$ and $t\geq 0$. 
    If $v\in\mathcal{M}_{\geqslant 0}(X)$ and $H\subseteq L$, one has the identity
\[
     R^{}_{H} \Bigl(\sum_{G\subseteq L} c^{(L)}_{G} (t) R^{}_{G} (v) \Bigr) 
     \, =  \sum_{G\subseteq L} c^{(L)}_{G} (t) R^{}_{G\cup H} (v) \ts ,
\]     
     which is then satisfied for all $ t\ge 0$.   \qed
\end{theorem}
Here, the upper index specifies the respective set of links. 
Clearly, the coefficient functions
$a^{}_{G} (t)$ of Theorem~\ref{thm:sol} satisfy the conditions of 
Theorem~\ref{thm:linear}. The result then means that the recombinators 
act linearly along the solutions~\eqref{eq:recosol} of the recombination equation~\eqref{eq:ode}.
Theorem~\ref{thm:linear} thus has  the consequence that,
on $\mathcal{M}_{\geqslant 0}(X)$, the forward flow of $\eqref{eq:ode}$ commutes with all 
   recombinators, that is, $R^{}_G\circ \varphi^{}_t  = \varphi^{}_t \circ R^{}_G$ for all 
   $t\geq 0$ and all $G\subseteq L$.

But let us go one step further here.
The conventional approach to solve the 
recombination dynamics consists in transforming 
the type frequencies to certain functions (known as principal components)
that diagonalise the dynamics, see 
\cite{Benn54,Daws02, Lyub92} and references therein for more.
We will now show that, 
in continuous time,  they have a particularly
simple structure: they are given by certain correlation functions,
known as {\em linkage disequilibria} (LDE) in biology, which
play an important role in  applications.
They have a counterpart at the level of operators (on $\mathcal{M}_{\geqslant 0}(X)$).
Namely, let us define \emph{LDE operators} via
\begin{equation}\label{eq:T_G}
T_G := \sum_{{\text{\d{$H$}}} \supseteq G} (-1)^{|H\setminus G|} R_H, \qquad G\subseteq L, 
\end{equation}
where the underdot indicates the summation variable. Note that
$T_G$ maps $\mathcal{M}_{\geqslant 0}(X)$ into  $\mathcal{M}(X)$, the
set of \emph{signed} measures on $X$. 
Eq.~(\ref{eq:T_G}) leads to the inverse $R_G = \sum_{{\text{\d{$H$}}}
\supseteq G} T_H$ by the combinatorial M\"obius inversion formula, 
see \cite[Thm.~4.18]{Aign79}.
We then have

\begin{theorem}\label{thm:diag}
If $\omega_t$ is the solution \eqref{eq:recosol}, 
the transformed quantities $T_G(\omega_t)$ satisfy
\begin{equation}\label{eq:TG_decay}
\frac{d}{dt} T_G (\omega_t) = - \Big ( \sum_{\alpha \in L \setminus G} 
\varrho^{}_\alpha \Big ) 
T_G (\omega_t), \qquad G\subseteq L. 
\end{equation}
\end{theorem}
\begin{proof}
See  \cite[Sec.~3.3]{BaBa03}.
\end{proof}

Obviously, Eq.~\eqref{eq:TG_decay} is a system of \emph{decoupled, linear,
  homogeneous} differential equations with the usual exponential solution. 
Note that this simple form emerged through the \emph{nonlinear} transform~(\ref{eq:T_G}) as applied to the solution of the \emph{coupled, nonlinear} 
differential equation \eqref{eq:ode}.

Suitable components of the signed measure $T_G(\omega_t)$ may then be identified
to work with in practice (see \cite{BaBa03,BaHe08} for details);
they correspond to  correlation functions
of all orders  and  decouple and decay 
exponentially. These functions turn out to be particularly well-adapted 
to the problem since they rely
on ordered partitions, in contrast to conventional LDE's used elsewhere
in population genetics, 
which  rely on general partitions (see \cite[Ch.~V.4]{Buer00}
for review).

\section{Stochastic dynamics, continuous time}
\label{sec:stoch_cont}

\subsection{The model}

The effect of finite population size in population genetics is,
in continuous time, well captured by the \emph{Moran model}. It describes
a population of fixed size $N$ and
takes into account the stochastic fluctuations due to random reproduction,
which manifest themselves via a \emph{resampling effect} 
(known as genetic drift in biology). 
More precisely,
the finite-population  counterpart of our deterministic model 
is the Moran model with single-crossover recombination. 
To simplify matters (and in order to clearly dissect the 
individual effects of recombination and resampling),
we shall use the decoupled (or parallel) version of the model, which assumes
that resampling  and recombination occur independently of
each other, as illustrated in Fig.~\ref{fig:moran}. More precisely, in our
finite population of fixed size $N$,  every individual 
experiences, independently of the others, 
\begin{itemize}\itemsep-1ex
\item resampling  at rate $b/2$. The individual reproduces, 
the offspring inherits the parent's type and replaces a
randomly chosen individual (possibly its own parent).\\
\item recombination  at (overall) rate $\varrho^{}_{\alpha}$ at link $\alpha \in L$. 
Every individual picks  a random partner (maybe itself) at rate 
$\varrho^{}_{\alpha}/2$, and the pair exchanges
the sites after link $\alpha$. That is,  if the recombining individuals
have types $x$ and $y$, they are replaced by the two offspring
individuals
$(x^{}_{<\alpha},y^{}_{>\alpha})$ and $(y^{}_{<\alpha},x^{}_{>\alpha})$, as in 
the deterministic case, and Fig.~\ref{fig:reco}. 
As before,
the per-capita rate of
recombination at link $\alpha$ is then $\varrho^{}_{\alpha}$, because
both orderings of the individuals lead to the same type count in the
population.
\end{itemize}

\begin{figure}[h]
  \begin{center}
  \includegraphics[scale=.6]{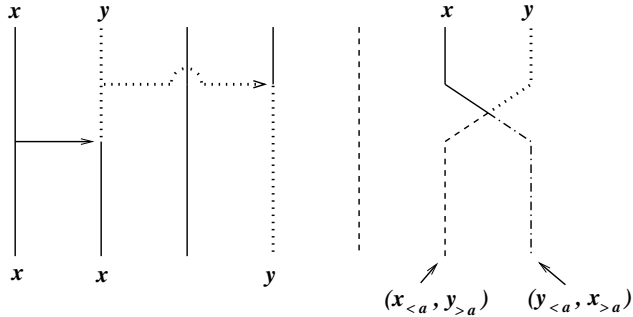}
  \end{center}  \caption{\label{fig:moran} {\small Graphical representation
of the Moran model (with parallel resampling and recombination). Every
individual is represented by a vertical line; time runs down the page.
Resampling is indicated by arrows, with the parent individual at
the tail and the offspring at the tip. Recombination is depicted by
a crossing between two individuals. Note that the spatial information
suggested by the graphical representation does not play a role in the model;
one is only interested in the frequencies of the various types.}}
\end{figure}

Note that the randomly chosen second individual
(for resampling or recombination)
may be the active individual itself; then, effectively, nothing happens.
One might, for biological reasons, prefer to exclude these events by
sampling from the remaining population only; but this means nothing but
a change of time scale of order $1/N$. 

To formalise this verbal description of the process,
let the state of the population at time $t$ be given by the collection
(the random vector)
\[
Z_t = \big(Z_t(x)\big)_{x\in X} \in 
E:= \Big \{z \in \{0,1,...,N\}^{|X|}_{} \,  \Big  | \, \sum_{x} z(x) = N \Big \},
\]
where
$Z_t(x)$ is the number of individuals of type $x$ at time $t$; clearly, 
$\sum_{x\in X}Z_t(x)=N$. We  also use $Z_t$ in the sense of a
(random counting) measure, in analogy with $\omega_t$ (but keep
in mind that $Z_t$ is integer-valued and counts single individuals,
whereas $\omega_t$
denotes continuous frequencies in an infinite population).
The letter $z$ will be used to denote realisations of $Z_t$ --- but 
note that the symbols $x,y$, and $z$ are not on equal
footing ($x$ and $y$ will continue to be types).
The stochastic process $\{Z_t\}_{t \geq 0}$ is the continuous-time Markov
chain on $E$ defined as follows.
If the current state is $Z_t=z$, two types of  transitions may occur:
\begin{eqnarray} 
&\text{resampling:} & z \rightarrow z + s(x,y), \quad 
                      s(x,y):= \delta_x - \delta_y  \nonumber, \\
 && \text{at rate } \frac{1}{2N} b \ts z(x) z(y) \; \text{ for }  
    (x,y) \in X \times X \label{eq:rate_resampling}\\
&\text{recombination:} & z \rightarrow z +r(x,y,\alpha), \nonumber \\
&                      & r(x,y,\alpha) :=  
\delta_{(x^{}_{<\alpha},y^{}_{>\alpha})} +\delta_{(y^{}_{<\alpha},x^{}_{>\alpha})} - \delta_x - \delta_y, \nonumber \\
 &&\text{at rate } \frac{1}{2N}
    \varrho^{}_{\alpha}z(x) z(y) \; \text{ for } 
  (x,y) \in X \times X, \alpha\in L \label{eq:rate_reco}
\end{eqnarray}
(where $\delta_x$ is the point measure on $x$, as before).
Note that, in \eqref{eq:rate_resampling}
and \eqref{eq:rate_reco}, transitions that leave $E$ are 
automatically excluded by the fact that  the corresponding rates vanish.
On the other hand, `empty transitions' ($s(x,y)=0$ or $r(x,y,\alpha)=0$) are 
explicitly included 
(they occur  if $x=y$ in  resampling
or recombination, and if $x^{}_{<\alpha}=y^{}_{<\alpha}$ or
$x^{}_{>\alpha}=y^{}_{>\alpha}$ in recombination).

\subsection{Connecting stochastic and deterministic models}

Let us now explore the connection between the stochastic process 
$\{Z_t\}_{t\geq 0}$ on $E$, its normalised version 
$\{\widehat Z_t\}_{t\geq 0}=\{Z_t\}_{t\geq 0}/N$ on $E/N$,
and the solution $\omega_t=\varphi_t(\omega_0)$ (Eq.~\eqref{eq:recosol}) of
the differential equation. It is easy to see
(and no surprise) that
\begin{equation}\label{eq:ivp}
\frac{d}{dt}\mathbb{E}(Z_t) = \EE \big(\varPhi(Z_t)\big), 
\end{equation}
with $\varPhi$ of \eqref{eq:ode}. But this does not,
\emph{per se}, lead to a `closed' 
differential equation for $\mathbb{E} (Z_t)$,  because 
it is not clear whether $\EE(\varPhi(Z_t))$ can be written as a function
of $\EE(Z_t)$ alone---after all, $\varPhi$ is nonlinear. 
\emph{In the absence of resampling,} however, we have

\begin{theorem}
\label{thm:reco_expect}
Let $\{Z_t\}_{t \geq 0}$ be the recombination process without resampling (i.e., $b=0$), and let $Z_0$ be fixed. Then,
$\mathbb{E} (Z_t )$ satisfies the differential equation
$$
\frac{d}{dt} \mathbb{E} (Z_t ) =  \varPhi \big ( \mathbb{E} (Z_t) \big )
$$
with initial value $Z_0$, and
$\varPhi$  from \eqref{eq:ode}; therefore,
\[
  \EE(Z_t) = \varphi_t(Z_0), \quad \text{for all} \; t \geq 0,
\]
with $\varphi_t$ from \eqref{eq:recosol}. Likewise, for all $t \geq 0$,
\[
  \EE(T_G Z_t) = T_G \big (\varphi_t(Z_0) \big ).
\]
\end{theorem}
\begin{proof}
See  \cite[Thm.~1 and Cor.~1]{BaHe08}. 
\end{proof}
The result again points to 
some underlying linearity, which, in the context of
the \emph{stochastic} model, should be connected to some kind of independence.
Indeed, the key to the proof of Thm.~\ref{thm:reco_expect} is
a lemma concerning the independence of marginal processes. 
For $I \subseteq S$, we introduce the `stretch' of $I$ as 
\[
  J(I) := \{i \in S \mid \min(I) \leq i \leq \max(I)\},
\]
and look at the projection of the recombination process on non-overlapping
stretches. This is the content of

\begin{lemma}\label{lem:indep_marginals} 
Let $\{Z_t\}_{t\geq 0}$ be the recombination process without resampling 
(i.e., $b=0$). Let $A,B \subseteq S$ with
$J(A) \cap J(B)=\varnothing$.
Then, $\{\pi^{}_A . Z_t\}_{t \geq 0}$ and $\{\pi^{}_B . Z_t\}_{t \geq 0}$
are conditionally (on $Z_0$) independent 
Markov chains on $E_A$ and $E_B$.
\end{lemma}
\begin{proof} 
See \cite[Lemma~1]{BaHe08}.
\end{proof}

Let us now re-include resampling, at rate $b/2>0$, and consider the
stochastic process $\{Z_t^{(N)}\}_{t \geq 0}$ defined by both 
\eqref{eq:rate_resampling}
and \eqref{eq:rate_reco}, where we add the upper index here to indicate
the dependence on $N$. Now, Lemma~\ref{lem:indep_marginals} and
Thm.~\ref{thm:reco_expect} are no longer valid.
The processes $\{\pi^{}_{<\alpha}.Z_t^{(N)}\}_{t \geq 0}$ 
and $\{\pi^{}_{>\alpha} . Z_t^{(N)}\}_{t \geq 0}$ are still individually Markov,
but their
resampling events are  coupled (replacement of $y^{}_{<\alpha}$ by $x^{}_{<\alpha}$
is always tied to replacement of $y^{}_{>\alpha}$ by $x^{}_{>\alpha}$).
Hence the marginal processes fail to be independent, so that
no equivalent of Lemma
\ref{lem:indep_marginals} holds.

Let us, therefore, change focus and consider the normalised version 
$\{\hZ_t^{(N)}\}_{t \geq 0} = \{Z_t^{(N)}\}_{t \geq 0}/N$.
In line with general folklore in population genetics,
in the limit $N\to \infty$, the 
relative frequencies  $\{\hZ_t^{(N)}\}_{t \geq 0}$
cease to fluctuate and are then given by the solution of the 
corresponding deterministic equation. More precisely, we have

\begin{prop}\label{prop:lln}
Consider the family of processes
$\{\hZ_t^{(N)}\}_{t \geq 0}=\frac{1}{N} \{Z_t^{(N)}\}_{t\geq 0}$,
$N=1,2,\ldots$,  where
$\{Z_t^{(N)}\}_{t\geq 0}$ is defined by 
\eqref{eq:rate_resampling} and \eqref{eq:rate_reco}.
Assume that the initial states are chosen so that
$\lim_{N\to\infty} \hZ^{(N)}_0 =p_0$. Then, 
for every given $t\geq 0$, one has
\begin{equation}\label{eq:LLN}
\lim_{N\to\infty} \sup_{s\leq t} | \hZ^{(N)}_s - p_s | =0 
\end{equation}
with probability $1$, where $p_s := \varphi_s(p_0)$ is the
solution of the deterministic recombination equation \eqref{eq:ode}. \qed
\end{prop}

The proof is an elementary application of Thm.~11.2.1 of \cite{EtKu86};
see Prop.~1 of \cite{BaHe08} for the explicit workout.

Note that the convergence in \eqref{eq:LLN} 
applies for any given $t$,
but need not carry over to $t \to \infty$. Indeed, if resampling is present, the population size
required to get close to the deterministic solution
is expected to grow over all bounds with increasing $t$. 
This is because, for every finite $N$, the Moran model with resampling and
recombination is an {\em absorbing} Markov chain, which leads to
fixation (i.e., to a homogeneous population of uniform type) in
finite time with probability one (for the special case of just two
types without recombination, the expected time 
is known to be of order
$N$ if the initial frequencies are both $1/2$ \cite[p.~93]{Ewen04}).
In sharp contrast, the
deterministic system never loses any type, and the stationary state,
the complete product measure \eqref{eq:stat_state},
is, in a sense, even the most variable state accessible to
the system. For increasing $N$, 
finite populations stay close to the deterministic limit for an increasing
length of time.

\section{Discrete time}
\label{sec:det_discr}
Let us  return to the deterministic setting and consider
the \emph{discrete-time} version of our single-crossover dynamics \eqref{eq:ode},
that is,
\begin{equation} \label{eq:reco_dis}
    \omega^{}_{t+1} \ts = \ts \omega^{}_t + \sum_{\alpha\in L} 
    \widetilde \varrho^{}_{\alpha} \bigl( R_{\alpha}
        - \one \bigr) ( \omega^{}_t )  =: \widetilde \varPhi(\omega^{}_t) \, .
\end{equation}
Here, the coefficients $\widetilde \varrho^{}_{\alpha}>0$, $\alpha\in L$, are the
\emph{probabilities} for a crossover at link $\alpha$ in every generation
(as opposed to the \emph{rates} $\varrho^{}_{\alpha}$ of the continuous-time
setting).
Consequently, we must
have  $0 < \sum_{\alpha\in L} \widetilde \varrho^{}_\alpha \leq 1$.

Based on the result for the continuous-time model,
the solution is expected to be of the form

\begin{equation}\label{solutiondiscrete}
   \omega^{}_t \,  = \, \widetilde \varPhi^t  (\omega^{}_0 ) 
       \ts  = \sum_{G\subseteq L} \widetilde a^{}_G ( t ) R^{}_G ( \omega^{} _0 ) \, ,
\end{equation}
with non-negative $\widetilde a^{}_G(t)$, $G\subseteq L$, 
$\sum_{G\subseteq L} \widetilde a^{}_G ( t ) = 1$,
describing the (still unknown) coefficient functions arising from the dynamics.
This representation of the solution was first stated by
Geiringer~\cite{Geir44}.
The coefficient functions will have the same probabilistic interpretation as the corresponding
$a^{}_G(t)$ in the continuous-time model,
so that $\widetilde a^{}_G(t)$ is the probability that the links that have been involved
in recombination until time $t$ are exactly those of the set $G$.

But there is a crucial difference.
Recall that, in continuous time, single cross\-overs imply
\emph{independence} of links, which is expressed in the product structure of the 
$a^{}_G(t)$ (see~Thm.~\ref{thm:sol}).
This independence is lost in discrete time, where a crossover event at one link 
forbids any other cut at other links in the same time step. It is
therefore not surprising that
a closed solution is not available in this case. It will, however, turn out
that a solution can be stated in terms of the (generalised) eigenvalues of the
system (which are known explicitly), together with coefficients to be determined
via a simple recursion. But it is rewarding to
take a closer look at the dynamics first.

Let us introduce the following abbreviations:
\[
   L^{}_{ \leq \alpha} := \left \{ i \in L \mid i \leq \alpha \right \} , \quad
    L^{}_{ \geq \alpha}  := \left \{ i \in L \mid i \geq \alpha \right \},
\]
and, for
each $G\subseteq L$,
\[
    G^{}_{< \alpha} := \left \{ i \in G \mid i < \alpha \right \} , \quad
    G^{}_{> \alpha}  :=  \left \{ i \in G \mid i > \alpha \right \}. 
\]   
Furthermore, we set 
$\eta := 1- \textstyle\sum_{\alpha \in L} \widetilde \varrho^{}_{\alpha}$.  
The dynamics \eqref{eq:reco_dis} is then reflected in the following
dynamics of the coefficient functions:

\begin{theorem}\label{thm:adevelop}
  For all $G\subseteq L$ and $t\in\mathbb{N}_0$, the coefficient functions 
$\widetilde a^{}_G(t)$ evolve according to
   \begin{equation}\label{nonlinrecur}
     \widetilde a^{}_G(t+1)\, = \, \eta \thinspace \widetilde a^{}_G(t) + \sum_{\alpha \in G } \widetilde \varrho^{}_{ \alpha} 
          \Bigl( \sum_{H \subseteq L_{ \geq \alpha} } \widetilde a^{}_{G^{}_{ < \alpha}\cup H}(t) \Bigr) \thinspace
          \Bigl( \sum_{K \subseteq L_{ \leq \alpha }} \widetilde a^{}_{K \cup G^{}_{> \alpha} }(t) \Bigr) \ts ,
   \end{equation}
with initial condition $\widetilde a^{}_G(0) = \delta^{}_{G,\varnothing}$. \qed
\end{theorem}

A verbal description of this dynamics was already given by Geiringer \cite{Geir44}; a formal proof may be found in \cite[Thm.~3]{WBB10}.

The above iteration is easily understood intuitively:
A type $x$ resulting from recombination at link $\alpha$ is composed of two segments $x^{}_{<\alpha}$
and $x^{}_{>\alpha}$. These segments themselves may have been pieced together in previous 
recombination events already, and the iteration explains the possible cuts these segments may
carry along. The first term in the product stands for the type delivering the leading segment (which
may bring along arbitrary cuts in the trailing segment),
the second for the type delivering the trailing one (here any leading segment is allowed).
The term $\eta \thinspace \widetilde a^{}_G(t)$ covers
the case of no recombination.

Let us now have a closer look at the structure of the dependence  between
links in  discrete time. To this end, note first that the set 
$G = \bigl\{\alpha^{}_1, \ldots, \alpha^{}_{\left|G\right|}\bigr\}\subseteq L$
with $\alpha^{}_1~< \alpha^{}_2 ~<~ \cdots~ <~ \alpha^{}_{\left|G\right|}$
partitions $L\setminus G$ into
$\mathcal{L}^{}_G := \left\{I_0^G, I_1^G, \ldots, I_{\left|G\right|}^G\right\}$, 
where
\begin{equation}\label{linksegments}
 \begin{split}
  I_0^G & = \left\{\alpha \in L : \tfrac{1}{2} \leq \alpha < \alpha^{}_1 \right\}, \quad
    I_{\left|G\right|}^G  = 
    \bigl\{\alpha \in L : \alpha^{}_{\vert G\vert} < \alpha \leq \tfrac{2n-1}{2} \bigr\} \, , \\
  \mbox{and }  I_{\ell}^G  &= \left\{\alpha \in L : \alpha^{}_{\ell} < \alpha < \alpha^{}_{\ell+1} \right\}
  \mbox{ for } 1\leq \ell\leq \vert G\vert -1. 
 \end{split}
\end{equation}
Cutting all links in $G$ decomposes the  original system (of sites and
links) into subsystems which are independent of each other from then on.
In particular, the links in $I_j$ become independent of those in
$I_k$, for $k \neq j$.  The probability that none of these 
subsystems experiences any  further recombination is
\begin{equation}\label{eq:def_lambda}
  \lambda^{}_G \,=\,  \prod_{ i = 0 }^{ \vert G \vert} 
            \Bigl( 1 - \sum_{ \alpha^{}_i \in I_i^G} \widetilde 
            \varrho^{}_{ \alpha^{}_i} \Bigr). 
\end{equation}
In particular, $\lambda^{}_{\varnothing} = \eta = 1-\sum_{\alpha\in L} \widetilde
\rho^{}_{\alpha}\geq 0$.
The $\lambda^{}_G$ are, at the same time, the generalised eigenvalues that
appear when the system is diagonalised and have been previously identified by
Bennett~\cite{Benn54}, Lyubich \cite{Lyub92} and Dawson~\cite{Daws02}.

A most instructive way to detail the effect of dependence
is the \emph{ancestral recombination process}:
start from an individual in the present population, let time
run backwards and consider how this individual's type has been
pieced together from different fragments in the past. In the four-sites
example of Fig.~\ref{fig:ancestral}, the probability 
that exactly $1/2$ and $3/2$
have been cut reads
\begin{equation}
\begin{split}\label{eq:a_twoway}
 \widetilde a^{}_{\{\frac{1}{2},\frac{3}{2}\}}(t)  = & \widetilde \varrho^{}_{\frac{1}{2}}\widetilde \varrho_{\frac{3}{2}} \sum_{k=0}^{t-2}
 \lambda_{\varnothing}^k \sum_{i=0}^{t-2-k} \lambda_{\frac{1}{2}}^i 
\lambda_{\{\frac{1}{2},\frac{3}{2}\}}^{t-2-k-i} \\
& + \widetilde \varrho^{}_{\frac{1}{2}}\widetilde \varrho^{}_{\frac{3}{2}} (1-\widetilde \varrho^{}_{\frac{5}{2}}) \sum_{k=0}^{t-2}
 \lambda_{\varnothing}^k \sum_{i=0}^{t-2-k} \lambda_{\frac{3}{2}}^i 
\lambda_{\{\frac{1}{2},\frac{3}{2}\}}^{t-2-k-i}.
\end{split}
\end{equation}
Here, the first (second) term corresponds to the possibility that link $1/2$ 
($3/2$) is the first to be
cut.
Obviously, the two possibilities are not symmetric: If $3/2$ is the first
to break, an additional factor of $(1-\widetilde \varrho_{5/2})$ is required to guarantee
that, at the time of the second recombination event (at $1/2$),
the trailing segment (sites 2 and 3) remains
intact while the leading segment (sites 0 and 1) is cut.

\begin{figure}[h]
  \psfrag{lk}{$\scriptstyle \lambda_{\varnothing}^k$}
  \psfrag{l1}{$\scriptstyle \lambda_{\{\frac{1}{2},\frac{3}{2}\}}^{t-2-k-i}$}
  \psfrag{l2}{$\scriptstyle \lambda_{\frac{1}{2}}^i$}
  \psfrag{l3}{$\scriptstyle \lambda_{\frac{3}{2}}^i$}
  \psfrag{o}{$\scriptstyle 1$}
  \psfrag{r1}{$\scriptstyle \widetilde \varrho_{\frac{1}{2}}$}
  \psfrag{r2}{$\scriptstyle \widetilde \varrho_{\frac{3}{2}}$}
  \psfrag{r3}{\mbox{\hspace{3mm}}
             $\scriptstyle 1-\widetilde\varrho_{\frac{5}{2}}$}
  \begin{center} \hspace{-1cm}
  \includegraphics[scale=.42]{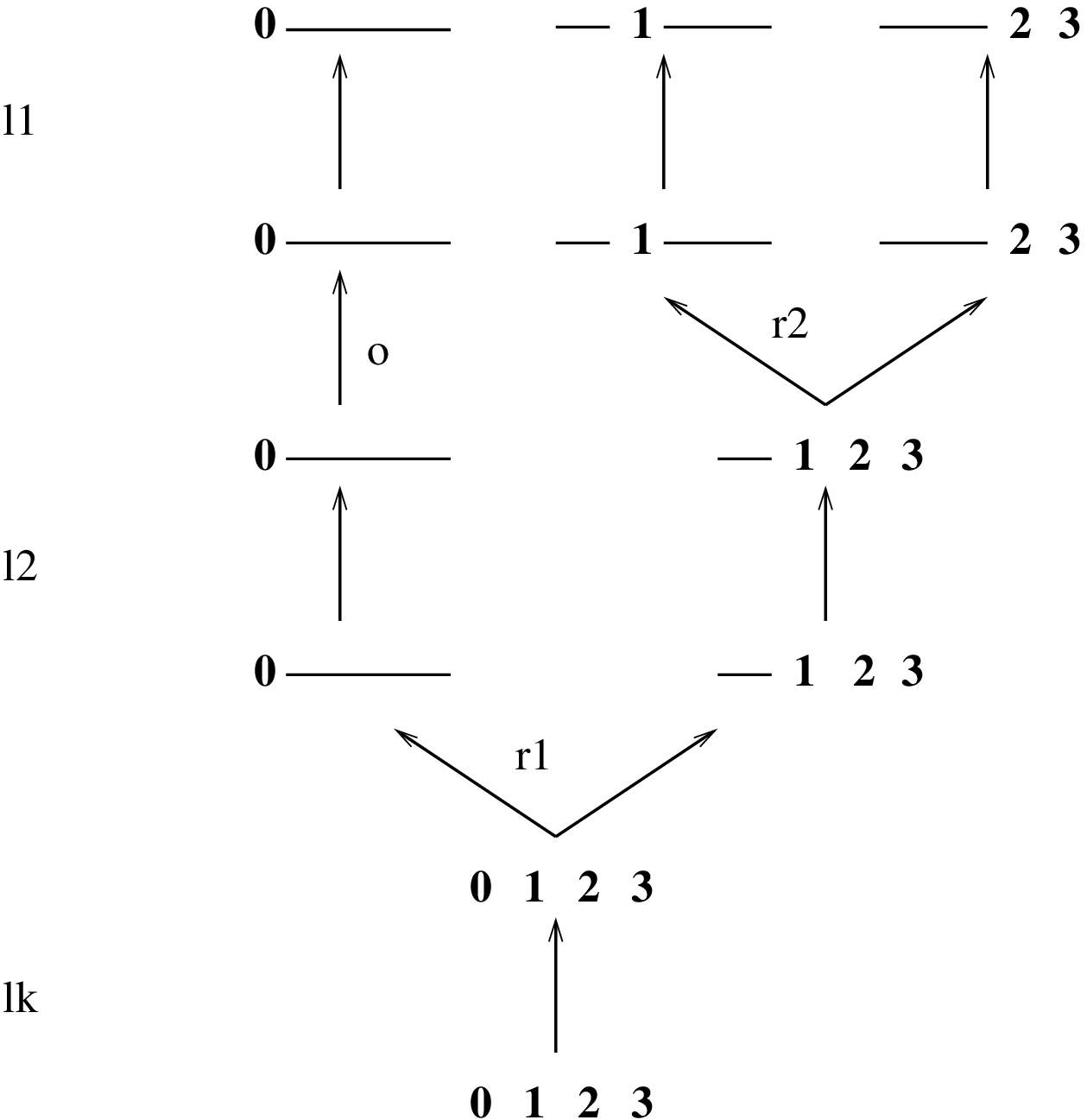}
  \qquad \quad
  \includegraphics[scale=.42]{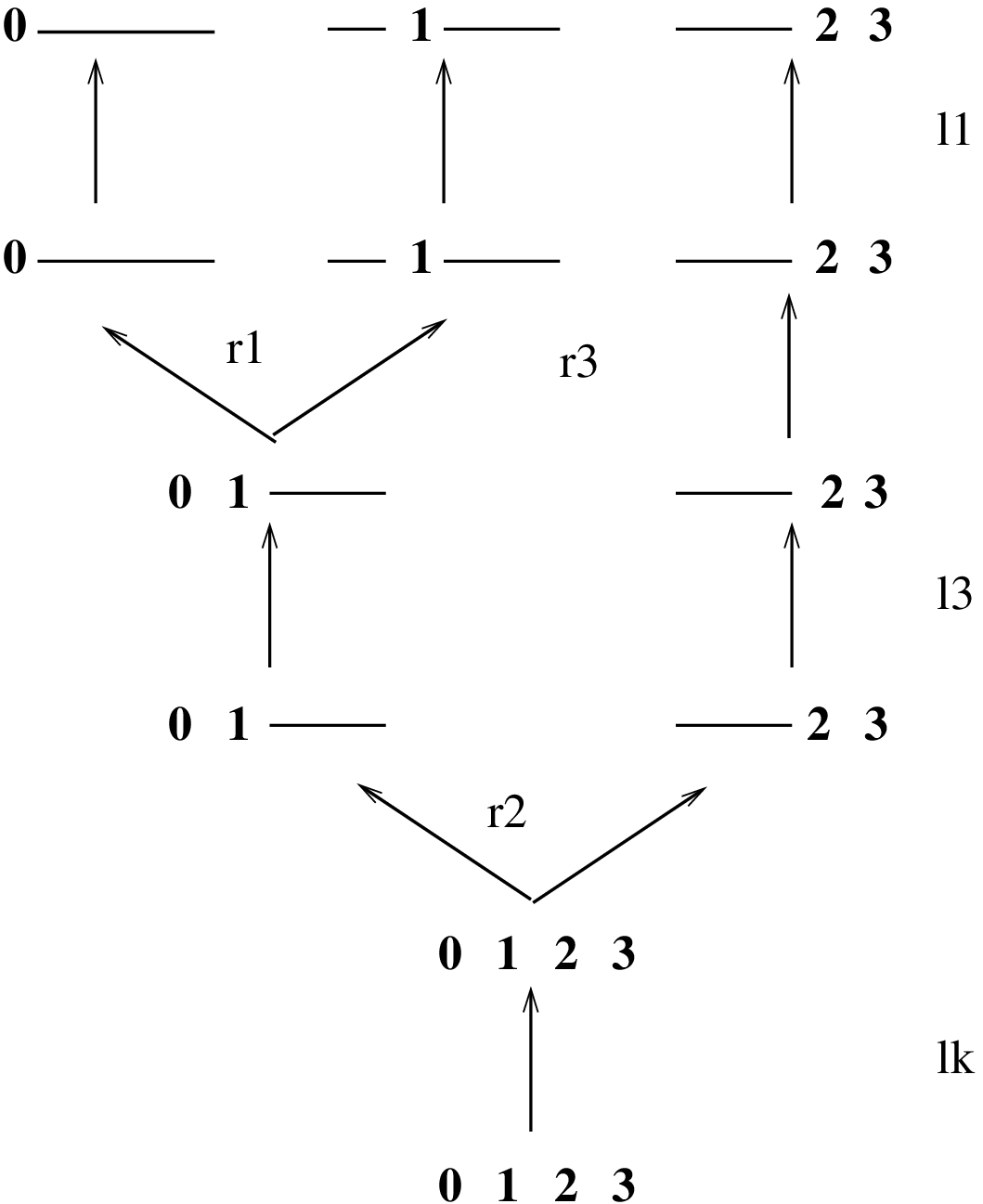}
  \end{center}  \caption{\label{fig:ancestral} {\small The ancestral
recombination process: possible histories of the sequence $0123$
(at the bottom). The two panels illustrate the
two terms of $\widetilde a_{\{1/2,3/2\}}(t)^{}$ in Eq.~\eqref{eq:a_twoway}}
(left: link $1/2$ is cut first; right: link $3/2$ is cut first.)
Arrows point in the backward direction of time. Blank lines indicate
arbitrary leading or trailing segments with which parts of the sequence 
have joined 
during recombination (they correspond to the asterisks ($*$) in 
Fig.~\ref{fig:reco}). The probability that nothing happens for a while
(straight arrows
only) is given by (powers of) the generalised eigenvalues \eqref{eq:def_lambda}.} 
\end{figure} 

Despite these complications, the discrete-time dynamics can
again be solved, even directly at the level of the $\widetilde a_G^{}(t)$, 
albeit slightly less explicitly than in continuous time. Indeed, it may
be shown (and will be detailed in a forthcoming paper) that
the coefficient functions have
the form
\[
 \widetilde a_G^{(L)}(t) = 
\sum_{H \subseteq G} \gamma^{(L)}_G(H) (\lambda_H^{(L)})^t,
\]
where the  upper index has again  been added to indicate the
dependence on the system.
The coefficents $\gamma^{(L)}_G(H)$ ($H \subseteq G$)
are defined recursively as follows. For $G \neq \varnothing$, 
\begin{equation}
\gamma^{(L)}_G(H) = \begin{cases}
  \frac{1}{\lambda_H^{(L)}-\lambda_{\varnothing}^{(L)}} \sum_{\alpha \in H}
  \widetilde \varrho^{}_{\alpha} \gamma^{(L<\alpha)}_{G_{<\alpha}}(H_{<\alpha}) 
  \gamma^{(L>\alpha)}_{G_{>\alpha}}(H_{>\alpha}), & H \neq \varnothing \\
  - \sum_{\varnothing \neq J \subseteq G} \gamma^{(L)}_G(J), & H=\varnothing.
  \end{cases}
\end{equation} 
Together with the initial value $\gamma^{(L)}_{\varnothing}(\varnothing)=1$,
this may be solved recursively.

A diagonalisation of the system (analogous to that in Thm.~\ref{thm:diag})
may also be achieved via a related, albeit technically more involved
recursion \cite{WBB10}.

\section{Concluding remarks and outlook}
The  results presented here can naturally only represent a  narrow
segment from a large area with lively recent and current activities.
Let us close this contribution by mentioning some important further directions
in the context of  recombination.

Our restriction to single crossovers provided a starting
point with a uniquely transparent structure (mainly due to the
independence of links in continuous time). However, arbitrary recombination
patterns (which partition the set of links into two \emph{arbitrary} 
parts) can also be dealt with, as has been done for the deterministic case
in  \cite{Lyub92,Daws02}. The underlying
mathematical structure will be further investigated in a forthcoming paper,
for both the deterministic and the stochastic models.

Above,  genetic material was exchanged reciprocally
at recombination events, so that the length of each sequence
remains constant. But sequences may also shift relative to each other
before cutting and relinking (so-called \emph{unequal crossover}),
which entails changes in length, see  \cite{Baak08} and references therein
for more.

The most important aspect of modern population genetics is the
backward-in-time point of view. This is natural because evolution is
mainly a historical science and today's researchers try to infer the
past from samples of individuals taken from present-day populations.
We have hinted at this
with our version of an ancestral recombination process, but  would
like to emphasise that this is only a toy version.
The full version of this process also takes into account resampling (as in 
Sec.~\ref{sec:stoch_cont}, with $b>0$) and
aims at the  law of \emph{genealogies} of \emph{samples} from \emph{finite} populations.
This point of view was introduced by Hudson \cite{Huds83}.
The fundamental concept here is the
\emph{ancestral recombination graph}: 
a branching-coalescing graph, where branching (backwards in time)
comes about by recombination (as in Fig.~\ref{fig:ancestral}), 
but lines may also
coalesce where two individuals go back to a common ancestor
(this corresponds to a  resampling event forward in time). 
For recent introductions into this topic,  see  \cite[Chap.~3]{Durr08},
\cite[Chap.~5]{HSW05}, or \cite[Chap.~7]{Wake09}; these 
texts also contain overviews of how  recombination may be  
inferred from genomic datasets.

Last not least, recombination and resampling are but two of the various
processes that act on genes in populations. Further 
inclusion of mutation and/or selection leads to a wealth of 
challenging problems,
whose investigation has stimulated the exploration of new mathematical
structures, concepts, and methods; let us only mention
\cite{BES04},  \cite{PHW06}, and \cite{JeSo10} as recent examples.
This development is expected to continue and intensify in the years
to come -- not least because it concerns the processes that have
shaped present-day genomes.



\begin{thebibliography}{99}

\bibitem{Aign79}
M.~Aigner,
\emph{Combinatorial Theory}, Springer, Berlin, 1979. Reprint 1997.

\bibitem{Baak01}
E.~Baake, Mutation and recombination with tight linkage,
\emph{J.\ Math.\ Biol.} {\bf 42} (2001), 455--488.

\bibitem{Baak05}
M.~Baake, Recombination semigroups on measure spaces,
\emph{Monatsh.\ Math.} {\bf 146} (2005), 267--278 and 
{\bf 150} (2007), 83--84 (Addendum).
arXiv:math.CA/0506099.  

\bibitem{Baak08}
M.~Baake, Repeat distributions from unequal crossovers,
\emph{Banach Center Publ.} {\bf 80} (2008), 53--70.
arXiv:0803.1270.

\bibitem{BaBa03}
M.~Baake and E.~Baake,
An exactly solved model for mutation, recombination and selection,
\emph{Canad.\ J.\ Math.} {\bf 55} (2003), 3--41 and {\bf 60} (2008), 264--265 
(Erratum).
arXiv:math.CA/0210422. 

\bibitem{BaHe08}
E.\ Baake and I.\ Herms,
Single-crossover dynamics: Finite versus infinite populations,
\emph{Bull.\ Math.\ Biol.} {\bf 70} (2008), 603--624.
arXiv:q-bio/0612024.

\bibitem{BES04}
N.~H.~Barton, A.~M.~Etheridge, and A.~K.~Sturm,
Coalescence in a random background,
\emph{Ann.~Appl.~Prob.} {\bf 14} (2004), 754-785.
arXiv:math/0406174.

\bibitem{Benn54}
J.~Bennett,
On the theory of random mating,
\emph{Ann.\ Hum.\ Genet.} {\bf 18} (1954), 311--317.

\bibitem{Buer00}
R.~B{\"u}rger,
\emph{The Mathematical Theory of Selection, Recombination, and
  Mutation},
Wiley, Chichester, 2000.

\bibitem{Daws02}
K.~Dawson,
The evolution of a population under recombination: How to linearize
  the dynamics,
\emph{Linear Algebra Appl.} {\bf 348} (2002), 115--137.

\bibitem{Durr08}
R.~Durrett,
\emph{Probability Models for DNA Sequence Evolution},
2nd ed., Springer, New York, 2008.

\bibitem{EtKu86}
S.~N. Ethier and T.~G. Kurtz,
\emph{Markov Processes -- Characterization and Convergence},
Wiley, New York, 1986.
Reprint 2005.

\bibitem{Ewen04}
W.~Ewens,
\emph{Mathematical Population Genetics},
Springer, Berlin, 2nd ed., 2004.

\bibitem{Geir44}
H.~Geiringer,
On the probability theory of linkage in {M}endelian heredity,
\emph{Ann.\ Math.\ Statist.} {\bf 15} (1944), 25--57.

\bibitem{HSW05}
J.~Hein, M.~H.~Schierup, and C.~Wiuf,
\emph{Gene Genealogies, Variation and Evolution},
Oxford University Press, Oxford, 2005.

\bibitem{Huds83}
R.~R.~Hudson,
Properties of a neutral allele model with intragenic
recombination,
\emph{Theor.\ Pop.\ Biol.} {\bf 23} (1983), 183--201.

\bibitem{JeSo10}
P.A.~Jenkins and  Y.S.~Song,
An asymptotic sampling formula for the coalescent
with recombination, 
\emph{Ann. Appl. Prob.}, in press.

\bibitem{Lyub92}
Y.~Lyubich,
\emph{Mathematical Structures in Population Genetics},
Springer, New York, 1992.

\bibitem{McHR83}
D.~McHale and G.~Ringwood,
Haldane linearization of baric algebras,
\emph{J.\ London Math.\ Soc.} {\bf 28} (1983), 17--26.

\bibitem{PHW06}
P. Pfaffelhuber, B. Haubold, and A. Wakolbinger,
Approximate genealogies under genetic hitchhiking,
\emph{Genetics} {\bf 174} (2006), 1995-2008.
 
\bibitem{Popa07}
E.~Popa,
Some remarks on a nonlinear semigroup acting on positive measures,
in: O.~Carja and I.~I. Vrabie, eds., \emph{Applied Analysis and
  Differential Equations}, pp. 308--319, World Scientific, Singapore, 2007.

\bibitem{Ring85}
G.~Ringwood,
Hypergeometric algebras and {M}endelian genetics,
\emph{Nieuw Archief v.\ Wiskunde} {\bf 3} (1985), 69--83.

\bibitem{Wake09}
J.~Wakeley,
\emph{Coalescent Theory.}
Roberts, Greenwood Village, CO, 2009.

\bibitem{WBB10}
U.~von Wangenheim, E.~Baake, and M.~Baake,
Single-crossover dynamics in discrete time,
\emph{J.\ Math.\ Biol.}, in press.  arXiv:0906.1678.


\end{thebibliography}
\end{document}